\documentclass[a4paper, 12pt]{article}

\usepackage{geometry}

\usepackage[utf8]{inputenc} 
\usepackage[english]{babel} 
\usepackage{amsfonts} 
\usepackage{amsmath} 
\usepackage{amssymb} 
\usepackage{amsthm} 
\usepackage{dsfont} 
\usepackage[backref]{hyperref} 
\usepackage{makecell} 
\usepackage{mathtools} 
\usepackage{multirow} 
\usepackage{siunitx} 
\usepackage{tikz} 
\usetikzlibrary{automata, positioning} 
\usepackage{cleveref} 

\newtheorem{theorem}{Theorem}
\newtheorem{proposition}{Proposition}
\newtheorem{lemma}{Lemma}
\newtheorem{corollary}{Corollary}
\newtheorem{example}{Example}

\newtheorem{remark}{Remark}

\DeclareMathOperator{\supp}{supp}
\newcommand{\1}{\mathds{1}}
\renewcommand{\S}{\mathcal{S}}

\newcommand{\A}{\mathcal{A}}
\newcommand{\B}{\mathcal{B}}
\newcommand{\Z}{\mathcal{Z}}
\newcommand{\M}{\mathcal{M}}
\newcommand{\V}{\mathcal{V}}
\newcommand{\E}{\mathcal{E}}
\newcommand{\another}[1]{\tilde{#1}}
\newcommand{\modified}[1]{{#1}'}
\newcommand{\transition}{\delta}
\newcommand{\Last}{\textrm{Last}}
\newcommand{\Reach}{\textrm{Reach}}
\newcommand{\Exit}{\textrm{Exit}}
\newcommand{\reach}{\textrm{reach}}
\newcommand{\exit}{\textrm{exit}}
\newcommand{\stay}{\textrm{stay}}
\newcommand{\policy}{\sigma}

\newcommand{\PP}{\mathbb{P}}
\newcommand{\NN}{\mathbb{N}}
\newcommand{\RR}{\mathbb{R}}
\newcommand{\ZZ}{\mathbb{Z}}

\newcommand{\eps}{\varepsilon}
\newcommand{\defas}{\coloneqq}

\title{Limit-sure reachability for small memory policies in POMDPs is NP-complete}

\date{}

\author{
	Ali Asadi
	\and
	Krishnendu Chatterjee
	\and
	Raimundo Saona
	\and
	Ali Shafiee
}

\begin{document}
	
\maketitle

\begin{abstract}
    A standard model that arises in several applications in sequential decision-making is partially observable Markov decision processes (POMDPs) where a decision-making agent interacts with an uncertain environment.
    A basic objective in such POMDPs is the reachability objective, where given a target set of states, the goal is to eventually arrive at one of them.
    The limit-sure problem asks whether reachability can be ensured with probability arbitrarily close to 1.
    In general, the limit-sure reachability problem for POMDPs is undecidable.
    However, in many practical cases the most relevant question is the existence of policies with a small amount of memory.
    In this work, we study the limit-sure reachability problem for POMDPs with a fixed amount of memory.
    We establish that the computational complexity of the problem is NP-complete.
\end{abstract}

\section{Introduction}

\paragraph{MDPs and POMDPs.}
A standard model in sequential decision-making is Markov decision processes (MDPs)~\cite{bellman1957MarkovianDecisionProcess, howard1960DynamicProgrammingMarkov}, which represents dynamical systems with both nondeterministic and probabilistic behavior.
MDPs provide the framework to model and solve control and probabilistic planning and decision-making problems~\cite{filar1997CompetitiveMarkovDecision, puterman2014MarkovDecisionProcesses} where the nondeterminism represents the choice of the control actions for the controller (or agent) and the probabilistic behavior represents the stochastic response of the system to control actions.
In perfectly observable MDPs the controller observes the evolution of the states of the system precisely, whereas in partially observable MDPs (POMDPs) the state space is partitioned according to observations for the controller, i.e., the controller can only view the observation of the current state (the partition the state belongs to) and not the precise state~\cite{bertsekas1976DynamicProgrammingStochastic, papadimitriou1987ComplexityMarkovDecision}.
POMDPs are widely used in several applications, including computational biology~\cite{durbin1998BiologicalSequenceAnalysis}, speech processing~\cite{mohri1997FinitestateTransducersLanguage}, image processing~\cite{culik1997DigitalImagesFormal}, software verification~\cite{cerny2011QuantitativeSynthesisConcurrent}, robot planning \cite{kress-gazit2009TemporalLogicBasedReactiveMission, kaelbling1998PlanningActingPartially}, reinforcement learning~\cite{kaelbling1996ReinforcementLearningSurvey}.

\paragraph{Reachability objectives and computational problems.}
A basic and fundamental objective in POMDPs is the reachability objective.
Given a set of target states, the reachability objective requires that some target state is visited at least once.
A policy is a recipe that resolves the choice of control actions. 
The main computational problems for POMDPs with reachability objectives are: 
(a) the quantitative problem asks if, for a fixed $\lambda \in (0, 1)$, there exists a policy that ensures the reachability objective with probability at least $\lambda$;
(b) the qualitative problem has two variants: 
(i)~almost-sure winning problem asks if there exists a policy that ensures the reachability objective with probability~$1$; and 
(ii)~limit-sure winning problem asks whether, for every $\lambda < 1$, there exists a policy that ensures the reachability objective with probability at least $\lambda$ (i.e., ensuring the reachability objectives with probability arbitrarily close to~$1$).

\paragraph{Significance of qualitative problems.}
The qualitative problem of limit-sure winning is of great significance in several applications.
For example, in the analysis of randomized embedded schedulers~\cite{baruah1992CompetitivenessOnlineRealtime, chatterjee2013AutomatedAnalysisRealtime}, the important question is whether every thread progresses with probability arbitrarily close to~1.
Moreover, in applications where it might be sufficient that the correct behavior happens with probability at least $\lambda < 1$, the correct choice of the threshold $\lambda$ can still be challenging, due to simplifications and imprecisions introduced during modeling.
In cases where almost-sure winning cannot be ensured limit-sure winning provides the strongest guarantee as compared to quantitative problems.
Besides its importance in practical applications, almost-sure and limit-sure convergence, like convergence in expectation, is a fundamental concept in probability theory, 
and provides the strongest probabilistic guarantees~\cite{durrett2019ProbabilityTheoryExamples}.

\paragraph{Previous results.}
The quantitative analysis problem for POMDPs with reachability objectives is undecidable~\cite{paz1971IntroductionProbabilisticAutomata}, and the undecidability result even holds for any approximation~\cite{madani2003UndecidabilityProbabilisticPlanning}.
In contrast, the complexities of the qualitative analysis problems are as follows:
(a)~the almost-sure winning problem is EXPTIME-complete~\cite{chatterjee2010QualitativeAnalysisPartiallyObservable, baier2012ProbabilisticOautomata}; and 
(b)~the limit-sure winning problem is undecidable \cite{gimbert2010ProbabilisticAutomataFinite, chatterjee2010ProbabilisticAutomataInfinite}.

\paragraph{Small-memory policies.}
While the computational complexities for the general problems are very high (undecidable in several cases), the same computational questions restricted to policies with small or constant amount of memory are important.
This is an interesting theoretical question and is practically relevant as the existence of a small-sized controller is desirable in all applications.
The existence of small-memory policies for almost-sure winning was studied in~\cite{chatterjee2016SymbolicSATBasedAlgorithm}, and proved to be NP-complete.
However, the quantitative problem is ETR-complete~\cite{junges2018FiniteStateControllersPOMDPs, junges2021ComplexityReachabilityParametric}, even for memoryless policies, where ETR is the existential theory of the reals, and it is a major open question if ETR is in NP or not.
The complexity of the limit-sure problem with respect to small-memory policies has not been studied and is the focus of this work.

\paragraph{Our contributions. }
In contrast to perfect-observation MDPs where almost-sure winning coincides with limit-sure winning, we show that in POMDPs almost-sure winning is different from limit-sure winning, see Example~\ref{Example: Almost-Sure is not Limit-Sure}.
Our main contribution to the limit-sure winning problem for POMDPs with reachability objectives with respect to constant memory policies is to establish that the computational complexity is NP-complete.	
Table~\ref{Table: Summary of results} summarizes the complexity results.

\begin{table}[ht]
    \centering
    \begin{tabular}{|l|c|c|}
        \hline
        \multirow{2}{*}{Problem} & \multicolumn{2}{c|}{Policies} \\
        \cline{2-3}
        & General & Constant memory \\
        \hline
        \hline
        Almost-sure & EXPTIME & NP-complete \\
        \hline
        Limit-sure & Undecidable & \makecell{\textbf{NP-complete}\\ \Cref{Result: NP-completeness}, \Cref{Result: constant memory NP-complete}} \\
        \hline
        Quantitative & Undecidable & ETR-complete \\
        \hline
    \end{tabular}
    \caption{%
        Complexity of quantitative, almost-sure, and limit-sure problems for general and constant-memory policies.
        Our contribution is marked in bold.
    }
    \label{Table: Summary of results}
\end{table}

\paragraph{Technical contributions.}
While we establish the same computational complexity of NP-completeness for limit-sure winning as for almost-sure winning, there are significant technical challenges. 
For example, in general, the almost-sure problem is EXPTIME-complete whereas the limit-sure problem is undecidable, which highlights that they are different problems.
Under memoryless policies, if a policy is almost-sure winning, then playing the actions in its support uniformly at random is also almost-sure winning, so it suffices to guess the support of actions. 
In contrast, we show that memoryless policies that are witness of the limit-sure winning property are more refined:  first, they are functional policies; second, there is a notion of ranks over actions where rank $k$ actions are played with probability proportional to $\eps^k$.

\paragraph{Related works.}
The area of POMDPs with applications is a huge and active research area. 
POMDPs with reachability objectives have been considered in the probabilistic automata theory community~\cite{gimbert2010ProbabilisticAutomataFinite,chatterjee2010ProbabilisticAutomataInfinite,chatterjee2010QualitativeAnalysisPartiallyObservable,baier2012ProbabilisticOautomata} as well as in the probabilistic planning community~\cite{kress-gazit2009TemporalLogicBasedReactiveMission,kaelbling1998PlanningActingPartially}.

The contingent or strong planning considers probabilistic choices as an adversary and is different from the qualitative winning problems we consider.
The strong cyclic planning problem is EXPTIME-complete~\cite{kaelbling1998PlanningActingPartially} and is closer to the almost-sure winning problem, but there are subtle differences, see~\cite{chatterjee2016SymbolicSATBasedAlgorithm}.

The almost-sure winning problem is considerably different from limit-sure winning which is in general undecidable, and none of the previous approaches apply to the limit-sure winning problem under small-memory policies.

Similarly, the connection between small-memory policies for POMDPs and parametric Markov chains (pMCs) was established by~\cite{junges2018FiniteStateControllersPOMDPs} for quantitative reach-avoid objectives. 
Later,~\cite{junges2021ComplexityReachabilityParametric} proved that some qualitative reachability problems for pMCs are all in NP.
They qualitative problems they considered include almost-sure reachability but not limit-sure reachability.
Therefore, this line of work doe not apply to the limit-sure winning problem under small-memory policies either.

\section{Preliminaries}

\paragraph{Notation.}
For a positive integer $n$ we denote the set $\{1, 2, \ldots, n\}$ by $[n]$.
For a set $\A$, the set of probability distributions over $\A$ is denoted by $\Delta(\A)$.
The probability distribution that assigns probability one to an element $a \in \A$ is denoted by $\1[a]$.
The disjoint union of sets is denoted by $\sqcup$.
For convenience, we will exchange the roles of $\lambda$ and $1 - \eps$ depending on the context.

\paragraph{POMDPs.}
A Partially Observable Markov Decision Process (POMDP) is a tuple $P = (\S, \A, \transition, \Z, o, s_0)$ where:
\begin{itemize}
    \item 
        $\S$ is a finite set of states;
    \item
        $\A$ is a finite set of actions;
    \item
        $\transition \colon \S \times \A \to \Delta(\S)$ is a probabilistic transition function that, given a state $s$ and an action $a$, returns the probability distribution over the successor states, i.e., the transition probability from $s$ to $s'$ given $a$ is denoted by $\transition(s, a)(s')$;
    \item
        $\Z$ is a finite set of observations;
    \item
        $o \colon \S \to \Z$ is an observation function that maps every state to an observation which, for simplicity, and without loss of generality, we consider that $o$ is a deterministic function~\cite[Remark~1]{chatterjee2015QualitativeAnalysisPOMDPs};
    \item
        $s_o \in \S$ is the unique initial state;
\end{itemize}
If $|\Z| = 1$, then we call the POMDP a blind MDP since the controller receives no information through the observations.
In this case, we identify the blind MDP with the tuple $(\S, \A, \transition, s_0)$.
Similarly, if $|\A| = |\Z|= 1$, then we call the POMDP a Markov chain, which coincides with the classic definition, and identify it simply with the tuple $(\S, \transition, s_0)$.

\paragraph{Plays and cones.}
A play (or a path) in a POMDP is an infinite sequence $(s_0, a_0, s_1, a_1, \ldots)$ of states and actions such that, for all $i \ge 0$, we have $\transition(s_i, a_i)(s_{i + 1}) > 0$. 
For a finite prefix $\omega \in (\S \times \A)^*$ of a play, the cone given by $\omega$ is the set of plays with $\omega$ as the prefix, and the last state of $\omega$, or $s_0$ if $\omega$ is empty, is denoted by $\Last(\omega)$. 
For a finite prefix $\omega = (s_0, a_0, s_1, a_1, \ldots, s_n, a_n)$ the sequence of observations and actions associated with $\omega$ is denoted by $o(\omega) = (o(s_0), a_0, o(s_1), a_1, \ldots, o(s_n), a_n) \in (\Z \times \A)^*$, which we call an observable history.

\paragraph{Policies.}
A policy is a recipe to extend prefixes of plays and is a function $\sigma \colon \Z \times (\A \times \Z)^* \to \Delta(\A)$ that, given a finite observable history, selects a probability distribution over the actions.
The set of all policies is denoted by $\Sigma$.

\paragraph{Policy with memory.}
A policy with memory is a tuple $\sigma = (\sigma_a, \sigma_u, \M, m_0)$ where: 
(i) $\M$ is a finite set of memory states; 
(ii) the function $\sigma_a \colon \M \times \Z \to \Delta(\A)$ is the action selection function that, given the current memory state and observation, gives the probability distribution over actions;
(iii) the function $\sigma_u \colon \M \times \Z \times \A \to \M$ is the memory update function that, given the current memory state, observation, and action, updates the memory state; and 
(iv) the memory state $m_0 \in \M$ is the initial memory state.
The set of all policies with memory amount $m$ is denoted by $\Sigma_m$.

\paragraph{Memoryless policies.}
A policy $\sigma$ is memoryless (or observation-stationary) if it depends only on the current observation, i.e., for every two histories $\omega$ and $\omega'$, if $o(\Last(\omega)) = o(\Last(\omega'))$, then $\sigma(o(\omega)) = \sigma(o(\omega'))$. 
Therefore, a memoryless policy is just a mapping from observations to a distribution over actions $\sigma \colon \Z \to \Delta(\A)$. 
The set of all memoryless policies corresponds to $\Sigma_1$.

\paragraph{Probability measure.}
Given a policy $\sigma$ and a starting state $s_0$, the unique probability measure over Borel sets of infinite plays obtained given $\sigma$ is denoted by $\PP_{s_0}^\sigma( \cdot )$, which is defined by Carathéodory's extension theorem by extending the natural definition over cones of plays~\cite{billingsley2012ProbabilityMeasurea}. 

\paragraph{Reachability objective and value.}
Given a set of target states, the reachability objective requires that a target state is visited at least once.
For simplicity and w.l.o.g., we consider that there is a single target state $\top \in \S$ since we can always add an additional state with transitions from all target states.
Formally, given a target state $\top \in \S$, the reachability objective is $\Reach(\top) = \{ (s_i, a_i)_{i \ge 0} \in (\S \times \A)^\NN \mid \exists i \ge 0 : s_i = \top \}$. 
The reachability value under $\Sigma$ is $\sup_{\sigma \in \Sigma} \PP_{s_0}^\sigma(\Reach(\top))$. 

\paragraph{Almost-sure winning.}
A POMDP $P$ with reachability objective $\Reach(\top)$ is almost-sure winning under $\Sigma$ if there is a fixed policy $\sigma \in \Sigma$ such that 
\[
\PP_{s_0}^\sigma(\Reach(\top)) = 1\,.
\]

\paragraph{Limit-sure winning.}
A POMDP $P$ with reachability objective $\Reach(\top)$ is limit-sure winning under $\Sigma$ if its reachability value under $\Sigma$ is $1$, i.e., if, for all $\eps > 0$, there is a policy $\sigma_\eps \in \Sigma$ such that $\PP_{s_0}^{\sigma_\eps} (\Reach(\top)) \ge 1 - \eps$, or equivalently, if 
\[
\sup_{\sigma \in \Sigma} \PP_{s_0}^\sigma(\Reach(\top)) = 1\,.
\]

\paragraph{Problems under constant memory.}
The limit-sure (resp.\ almost-sure) problem under constant amount of memory $m \ge 1$ asks whether a POMDP $P$ is limit-sure (resp.\ almost-sure) winning under policies restricted to~$\Sigma_m$.

\section{Computational Complexity}

In this section, we present the main complexity result and show that almost-sure winning and limit-sure winning are different properties in POMDPs through the following example.

\begin{example}[Almost-sure $\not =$ Limit-sure]
    \label{Example: Almost-Sure is not Limit-Sure}
    Consider a blind MDP (with no helpful observation) with four states and two actions \emph{wait}, \emph{w}, and \emph{commit}, \emph{c}. 
    The transitions are such that, under action \emph{wait}, the initial state, $s_0$, may loop with positive probability or may transition to a second state, $s_1$, with positive probability.
    Under action \emph{commit}, the initial state moves to an absorbing state, $\bot$, while the second state reaches the target, $\top$. 
    See Figure~\ref{Figure: Limit-sure example} for an illustration.
    The reachability value of this blind MDP is $1$.
    On the one hand, there is no policy that guarantees reachability value one, and therefore the blind MDP is not almost-sure winning.
    On the other hand, for every $\eps > 0$, a policy guaranteeing a reachability probability of at least $1 - \eps$ requires playing action \emph{wait} sufficiently many times before playing action \emph{commit}.
    For every $\eps > 0$, this behavior can be simulated by a distribution over actions that assigns little probability to action \emph{commit}. 
    Therefore, the blind MDP is limit-sure winning, even under memoryless policies. \qed
\end{example}

\begin{figure}[ht]
    \centering
    \begin{tikzpicture} [
        node distance = 1cm, 
        every initial by arrow/.style = {thick}
        ]
        
        \node (s0) [state] {$s_0$};
        \node (s1) [state, right = of s0] {$s_1$};
        \node (s2) [state, left = of s0] {$\bot$};
        \node (top) [state, right = of s1] {$\top$};
        
        \path [-stealth, thick]
        (s0) edge[loop above] node {$w$}   ()
        (s0) edge[bend left, above] node {$w$}   (s1)
        (s1) edge [loop above]  node {$w$} ()
        (s0) edge [above]  node {$c$} (s2)
        (s1) edge [above]  node {$c$} (top)
        (s2) edge [loop above]  node {$w, c$}()
        (top) edge [loop above]  node {$w, c$}();
    \end{tikzpicture}
    \caption{
        Example of POMDP that is limit-sure winning but not almost-sure winning.
        Edges represent a positive probability transition between states when the corresponding action in its label is used.
    }
    \label{Figure: Limit-sure example}
\end{figure}

\paragraph{Main novelty of limit-sure vs almost-sure winning.}
The limit-sure winning property relates to a sequence of policies, as opposed to the almost-sure winning property which relates to a single policy. 
Moreover, given a sequence of policies $(\policy_\eps)_{\eps > 0}$ that prove the limit-sure winning property, if it exists, the limit policy $\lim_{\eps \to 0^+} \policy_\eps$ is not a witness of the limit-sure winning property.
This is the case in \Cref{Example: Almost-Sure is not Limit-Sure} where the limit policy applies action \emph{wait} always and does not indicate that the POMDP is limit-sure winning. 
To preserve the asymptotic information, we work with symbolic or functional policies, called \emph{rank policies}, which assign probabilities to actions based on ranks. 
For intuition, lower ranks have higher priority, and, if low-rank actions form a cycle, then higher-rank actions determine the exit distribution. 
In \Cref{Example: Almost-Sure is not Limit-Sure}, a rank policy giving a low rank to action \emph{wait} and a high rank to \emph{commit} reflects that the POMDP is limit-sure winning.
Note that considering rank policies with only one rank is equivalent to classic policies that choose actions uniformly at random in its support, which is enough to solve the almost-sure problem.

We now state the main complexity result.
\begin{theorem}
    \label{Result: NP-completeness}
    The problem of determining whether a POMDP $P$ with reachability objective is limit-sure winning under memoryless policies is NP-complete.
\end{theorem}

The rest of this section is dedicated to the proof of \Cref{Result: NP-completeness}.
First, we recall some fundamental concepts.
Second, we show the NP upper bound by proving the existence of rank policy witnesses of polynomial size and providing a polynomial-time verifier.
Third, we show the NP-hardness by a reduction from 3-SAT.
Finally, we present extensions of \Cref{Result: NP-completeness} for small-memory policies and objectives other than reachability.

\subsection{Previous Concepts}

We introduce the most important previous concepts used in our proof.

\paragraph{Real-closed fields.}
A real-closed field $R$ is a field, i.e., a set on which addition, subtraction, multiplication, and division work as usual, and moreover the intermediate value theorem applies.
For an introduction, see~\cite{basu2006AlgorithmsRealAlgebraic}.

\paragraph{Puiseux functions.}
The set of Puiseux functions is the set of functions $f \colon [0, \eps_0) \to \RR$ of the form $f(\eps) = \sum_{i \ge k} c_i \eps^{i/q}$ where $k \in \ZZ$, $i$ ranges in $\ZZ$, $c_i \in \RR$, $q \in \NN$, and $\eps_0  > 0$.
The field of Puiseux functions is an important example of a nonarchimedean real-closed field.

\begin{theorem}[{\cite[Section 10]{bewley1976AsymptoticTheoryStochastic}}]
    \label{Result: Puiseux functions}
    The field of Puiseux functions is real-closed.
\end{theorem}

\paragraph{First order theory of the reals.}
A sentence in the first order theory of the reals $\phi$ is given by $Q_1 x_1 Q_2 x_2 \ldots Q_k x_k$ $F(x_1, x_2 \ldots, x_k)$, where $Q_i \in \{ \exists, \forall\}$ are quantifiers and $F(x_1, x_2 \ldots, x_k)$ is a quantifier-free formula in the language of ordered fields with coefficients in a real-closed field.
The decision problem for the first-order theory of the reals is, given a sentence $\phi$, to decide whether it is true or false.
A fundamental result in logic is the following.

\begin{theorem}[{Tarski-Seidenberg principle~\cite[Theorem 2.80, page 70]{basu2006AlgorithmsRealAlgebraic}}]
    \label{Result: Tarki's principle}
    Suppose that $R$ is a real-closed field that contains the real-closed field $\RR$. 
    If $\phi$ is a sentence in the language of ordered fields with coefficients in $\RR$, then it is true in $\RR$ if and only if it is true in $R$.
\end{theorem}

The following result is a characterization of the reachability value in Markov chains.

\begin{theorem}[{\cite[Theorem 10.15, page 762]{baier2008PrinciplesModelChecking}}]
    \label{Result: value characterization}
    Consider a Markov chain with a set of states $\S$ and a target set $\{ \top \}$.
    Then, the reachability value as a function of the initial state is a solution $v^* \in [0, 1]^{\S}$ of the system of equations given by $v(\top) = 1$ and, for all $s \in \S \setminus \{\top\}$,
    \[
    v(s) = \sum_{\another{s} \in \S} \delta(s, \another{s}) v(\another{s}) \,,
    \]
    such that, for all other solutions $u^*$, we have that $v^*(s) \le u^*(s)$, for all $s \in \S$.
\end{theorem}

Since we consider Markov chains whose transition probabilities are Puiseux functions, which we call Puiseux Markov chains, we introduce a few concepts used in~\cite{solan2003PerturbationsMarkovChains}.

\paragraph{Puiseux Markov chains.}
A Puiseux Markov chain is a family of Markov chains parameterized by $\eps$ where the transition function is a Puiseux function $\eps \mapsto \delta^\eps \colon \S \to \Delta(\S)$.
In particular, for each $\eps$, the transition $\delta^\eps$ and the starting state $s$ induces a probability measure $\PP^\eps_s$.

\paragraph{Reach and exit times, and exit event.}
For a Markov chain, a state $s \in \S$ and a set of states $\B \subseteq \S$, we consider the following random variables:
\begin{align*}
	\exit(\B) 
	&\defas \min \{ n \ge 0 : s_n \not \in \B \} \,,\\
	\reach(s) 
	&\defas \min \{ n \ge 0 : s_n = s \} \,, \text{ and}\\
	\Exit(\B, s) 
	&\defas \{ \exit(\B) < \infty \} \cap \{ s_{\exit(\B)} = s \} \,,
\end{align*}
i.e., $\exit(\B)$ is the first time a state outside of $\B$ is visited, $\reach(s)$ is the first time the state $s$ is visited, and $\Exit(\B, s)$ is the event of exiting the set $\B$ by visiting state $s$. 
In particular, the event $\Reach(\another{s})$ is equivalent to $\reach(\another{s}) < \infty$.

The following definition generalizes the concept of communicating class in Markov chains for Puiseux Markov chains.

\paragraph{Communicating class in Puiseux Markov chains.}
Given a Puiseux Markov chain, a set of states $\B \subseteq \S$ is a communicating class if, for all $s, \another{s} \in \B$, we have
\begin{align*}
	\lim_{\eps \to 0^+} \PP_{s}^\eps\bigl( \exit(\B) < \reach(\another{s}) \bigr) &= 0\,, \\
	\lim_{\eps \to 0^+} \PP_{s}^\eps\bigl( \Reach(\another{s}) \bigr) &= 1\,,
\end{align*}
i.e., starting from $s$, state $\another{s}$ is visited before exiting $\B$.
Note that the second condition corresponds to the case of $\exit(\B) = \reach(\another{s}) = \infty$ in Solan~\cite{solan2003PerturbationsMarkovChains}, which is implicitly included in this previous work and prevents that a communicating class consists of unconnected states.

The following concept is key to characterizing events in Puiseux Markov chains.

\paragraph{Exit graph.}
Given a Puiseux Markov chain and a set of states $\B \subseteq \S$, an exit graph of $\B$, denoted $g$, is a directed acyclic graph with edges $E(g) \subseteq \B \times \S$ such that, for all $s \in \B$, there exists $\another{s} \in \S$ such that $(s, \another{s}) \in E(g)$.
We denote the set of all exit graphs of $\B$ by $G_{\B}$, and all those in which $s$ can reach $\another{s}$ by $G_{\B}(s \to \another{s})$.
The probability of an exit graph $g$ is the product of the probability of each of its transitions defined as $\delta(g) \defas \prod_{(s, \another{s}) \in g} \delta(s)(\another{s})$.
The \emph{weight} of an exit graph $g$ is the leading power in the Puiseux series expansion of the product of the involved transitions defined as
\[
w(g) \defas \inf \left \{ r \ge 0 : \lim_{\eps \to 0^+} \frac{ \delta^\eps (g) }{\eps^r} \not = 0 \right \} \,. 
\]
The set of exit graphs of $\B$ that have minimal weight is denoted by $G^{\min}_{\B}$ and $G^{\min}_{\B}(s \to \another{s})$ for those in which $s$ can reach $\another{s}$.

The following result shows that the exit distribution of a communicating class is independent of the initial state within the communicating class. 

\begin{theorem}[{\cite[Lemma 3, page 270]{solan2003PerturbationsMarkovChains}}]
    \label{Result: exit distribution}
    Consider a  Puiseux Markov chain and a communicating class $\B \subseteq \S$.
    Then, the following expression is independent of the initial state $s \in \B$
    \[
    \delta(\B, \another{s}) \defas \lim_{\eps \to 0^+} \PP^\eps_{s}(\Exit(\B, \another{s})) \,.
    \]
\end{theorem}

Finally, the following result characterizes the exit ``distribution'' in terms of exit graphs.

\begin{theorem}[{\cite[Equation 6, page 268]{solan2003PerturbationsMarkovChains}}]
    \label{Result: characterization exit distribution}
    Consider a  Puiseux Markov chain and a communicating class $\B \subseteq \S$.
    Then, for all $s \in \B$ and $\another{s} \in \S \setminus \B$, 
    \[
    \delta(\B, \another{s}) 
    = \lim_{\eps \to 0^+} \frac{ \sum_{g \in G^{\min}_{\B}(s \to \another{s})} \delta^\eps(g) }{ \sum_{g \in G^{\min}_{\B}} \delta^\eps(g) } \,,
    \]
    where the sum over an empty set is $0$ and the quotient $0/0$ is also $0$.
\end{theorem}

We call $\delta( \B, \cdot )$ an exit distribution even when it can be constant to zero.
Its support corresponds to all the states mapped to a strictly positive value.
The following concept allows us to characterize limit-sure reachability in Puiseux Markov chains.

\paragraph{Absorbing communicating class.}
Given a Puiseux Markov chain, a communicating class $\B \subseteq \S$ is absorbing if its exit distribution has empty support, i.e., $\supp( \delta(\B, \cdot) ) = \emptyset$.

The following concept is classic in Graph theory and we introduce it for completeness.

\paragraph{Bottom strongly connected component of a directed graph.}
In a directed graph, a bottom strongly connected component is a set of states where: there is a directed path from every state to every other state in the set, and all edges starting in the set lead to states within the set.

\subsection{Upper Bound}
\label{Section: Upper bound}

The NP upper bound complexity is our main result and is established through the following sequence of results.
\begin{enumerate}
    \item 
    We establish a reduction from general POMDPs to blind MDPs (\Cref{Result: Reduction to blind MDPs}).
    \item 
    For blind MDPs, we establish the existence of Puiseux function policy witnesses (\Cref{Result: function witness}).
    \item 
    We establish the laminar structure of a graph of communicating classes in Puiseux Markov chains (\Cref{Result: Laminar communicating classes}).
    \item 
    We establish that the graph of communicating classes characterizes reachable states in Puiseux Markov chains (\Cref{Result: Characterization of reachable states}).
    \item 
    We establish properties of the graph of communicating classes that characterize limit-sure reachability in Puiseux Markov chains (\Cref{Result: Characterization of limit-sure reachability}).
    \item 
    We establish the existence of rank policy witnesses, a simple and polynomial-size policy (\Cref{Result: rank witness existence}).
    \item 
    We provide a polynomial-time verifier for rank policy witnesses (\Cref{Result: Verifier existence}).
\end{enumerate}

The following result establishes a reduction from general POMDPs to blind MDPs.

\begin{lemma}
    \label{Result: Reduction to blind MDPs}
    For every POMDP $P = (\S, \A, \delta, \Z, o, s_0)$, there exists a blind MDP with $\modified{P} = (\S, \A \times \Z, \modified{\delta}, s_0)$ with the same reachability value under memoryless policies.
\end{lemma}

\begin{proof}[Proof sketch]
    The action $(a, z)$ in the blind MDP corresponds to applying action $a$ only to states whose observation is $z$. 
    We define the transition $\modified{\delta}$ accordingly by introducing loops when an action $(a, z)$ is applied and the underlying state $s$ has a different observation, i.e., $z \not = o(s)$.
    A coupling on the underlying dynamics, that eliminates the introduced loops in the blind MDP, shows that the reachability values under memoryless policies coincide.
\end{proof}

\begin{proof}
    Consider an arbitrary POMDP $P = (\S, \A, \delta, \Z, o, s_0)$.
    Define a blind MDP $\modified{P} = (\S, \modified{\A}, \modified{\delta}, \modified{\Z}, \modified{o}, s_0)$ where
    \begin{itemize}
        \item 
        $\modified{\A} \defas \A \times \Z$;
        \item 
        $\modified{\delta} \colon \S \times \modified{\A} \to \Delta(\S)$ is given by
        \[
        \modified{\delta}(s, (a, z)) \defas \begin{cases}
            \delta(s, a)
            & o(s) = z \\
            \1[s]
            & o(s) \not = z 
        \end{cases}
        \]
        \item 
        $\modified{\Z} \defas \{ \# \}$ a unique observation;
        \item 
        $\modified{o} \equiv \#$ a uninformative observation function.
    \end{itemize}
    We show that the value of this blind MDP $\modified{P}$ is the same as the original POMDP $P$.
    
    Consider an arbitrary memoryless policy $\sigma \colon \Z \to \Delta(\A)$ in the POMDP $P$.
    Note that $\sigma \in \Delta(\A)^\Z$ is a collection of distributions.
    Define the memoryless policy $\modified{\sigma} \colon \modified{\Z} \to \Delta(\modified{\A})$ in $\modified{P}$, which we identify with an element of $\Delta(\modified{\A})$, by a uniform choice over the distributions in $\sigma$, i.e.,
    \[
    \modified{\sigma}((a, z)) \defas \frac{1}{|Z|} \sigma(z)(a) \,.
    \]
    In other words, the policy $\modified{\sigma}$ chooses an observation $z$ uniformly at random and then an action according to the distribution $\sigma(z)$.
    
    The coupling between the blind MDP and the POMDP consists in projecting the dynamic of the blind MDP to those times where the tuple of action and observation $(a, z)$ is such that $z$ is the observation of the current state.
    Formally, define a sequence of random times $(\tau_t)_{t \ge 0}$ defined inductively by $\tau_0 \defas \inf \{ t \ge 0 : \exists a \in \A \quad \modified{A}_t = (a, o(s_0)) \}$ and, for $t \ge 1$,
    \[
    \tau_t \defas \inf \{ t > \tau_{t - 1} : \exists a \in \A \quad \modified{A}_t = (a, o(S_{t - 1})) \} \,.
    \]
    In other words, $\tau_t$ is the $t$-th time that, in the dynamic of the blind MDP, the second coordinate of the action chosen by $\modified{\sigma}$ coincides with the observation of the current state in the original POMDP.
    These times are almost surely finite since $\modified{\sigma}$ chooses an observation uniformly at random at each step.
    Notice that, after coupling the transitions of $\sigma$ and $\modified{\sigma}$ in the obvious way, $(S_{\tau_t})_{t \ge 0}$ under $\modified{\sigma}$ and $(S_t)_{t \ge 0}$ under $\sigma$ follow the same dynamic.
    In particular, the probability of reaching the target is equal in the blind MDP and the POMDP.
    Therefore, the reachability value of the blind MDP is at least as large as the reachability value of the POMDP.
    
    Consider an arbitrary memoryless policy $\modified{\sigma} \colon \modified{\Z} \to \Delta(\modified{\A})$ in the blind MDP $\modified{P}$, or equivalently an element of $\Delta(\modified{\A})$.
    Define the memoryless policy $\sigma \colon \Z \to \Delta(\A)$ in $P$ by
    \[
    \sigma(z)(a) \defas \frac{\modified{\sigma}((a, z))}{ \sum_{\another{a} \in \A} \modified{\sigma}((\another{a}, z))} \,.
    \]
    In other words, for each observation, we consider the conditional distribution of $\modified{\sigma}$ on the actions that have that observation as a second coordinate.
    
    Just as before, the same coupling shows $(S_{\tau_t})_{t \ge 0}$ under $\modified{\sigma}$ and $(S_t)_{t \ge 0}$ under $\sigma$ follow the same dynamic.
    Therefore, the reachability value of the POMDP is at least as large as the reachability value of the blind MDP.
    We conclude that both POMDPs have the same value.
\end{proof}

\paragraph{Puiseux function policy.}
A (memoryless) Puiseux function policy $\policy$ is a function $\policy \colon [0, \eps_0) \to \Delta(\A)$.
Note that, for all $\eps \in [0, \eps_0)$, the policy $\policy(\eps)$ is a memoryless policy and, together with an initial state $s$, induces a Markov chain whose measure is denoted $\PP^{\policy(\eps)}_s$.
For example, for some $a \in \A$, we may have that $\policy(\eps)_a = 1 / (2 - \eps)$, which is a probability for $\eps \in [0, 1]$.

The following result establishes the existence of Puiseux function policy witnesses for blind MDPs.

\begin{lemma}
    \label{Result: function witness}
    Consider a blind MDP $P = (\S, \A, \delta, s_0)$ and a target state $\top \in \S$.
    Then, $P$ is limit-sure winning under memoryless policies if and only if the following decision problem for the first-order theory of the reals has a solution in the real-closed field of Puiseux functions.
    
    $\forall \lambda < 1$ 
    $\exists (\sigma_a)_{a \in \A}$ 
    $\exists (v_s)_{s \in \S}$ such that 
    \begin{itemize}
        \item 
        Policy: for all $a \in \A$, we have that $\sigma_a \ge 0$, and $\sum_{a \in \A} \sigma_a = 1$.
        \item 
        Fixpoint: for all $s \in \S$, we have that $v$ satisfies 
        \[
        v_s = \sum_{\another{s} \in \S} \sum_{a \in \A} \sigma_a \delta(s, a)(\another{s}) \, v_{\another{s}} \,.
        \]
        \item 
        Minimal solution:
        $\forall (u_s)_{s \in \S}$, if $u$ satisfies the previous fixpoint equation, then, for all $s \in \S$, $v_s \le u_s$.
        \item 
        Value:
        $v_{s_0} \ge \lambda$.
    \end{itemize}
\end{lemma}

\begin{proof}[Proof sketch]
    We follow an approach similar to~\cite{bewley1976AsymptoticTheoryStochastic} where we characterize the limit-sure winning problem as a decision problem in the first order theory of the reals: the POMDP is limit-sure winning if and only if this decision problem is true.
    By Tarski's principle~\cite[Theorem 2.80, page 70]{basu2006AlgorithmsRealAlgebraic}, the decision problem is true if and only if it has a witness in the field of Puiseux functions, which is a real-closed field.
    Therefore, limit-sure winning POMDPs have Puiseux functions policy witnesses.
\end{proof}

\begin{proof}
    Consider a blind MDP $P$.
    Recall that $P$ is limit-sure winning under memoryless policies if and only if
    \[
    \sup_{\sigma \in \Sigma_0} \PP_{s_0}^\sigma(\Reach(\top)) = 1 \,.
    \]
    In other words, if and only if, for all $\lambda < 1$ there exists a policy $\sigma_\lambda \in \Sigma_0$ such that $\PP_{s_0}^\sigma(\Reach(\top)) \ge \lambda$.
    Recall that, since $P$ is a blind MDP, the set of memoryless policies $\Sigma_0$ is equivalent to $\Delta(\A)$, so an element $\sigma \in \Sigma_0$ is fully determined by the probability it assigns to each action.
    
    By \Cref{Result: value characterization}, a policy $\sigma_\lambda \in \Sigma_0$ is such that $\PP_{s_0}^{\sigma_\lambda} (\Reach(\top)) \ge \lambda$ if and only if the corresponding Markov chain has a value vector such that $v_{s_0} = v(s_0) \ge \lambda$.
    So far, we conclude that $P$ is limit-sure winning under memoryless policies if and only if the stated decision problem for the first-order theory of the reals has a solution in $\RR$.
    By \Cref{Result: Tarki's principle}, since the Puiseux functions is a real-closed field by \Cref{Result: Puiseux functions} and it contains $\RR$, we conclude the proof.
\end{proof}

\paragraph{Graph of communicating classes.}
A memoryless Puiseux function policy $\sigma$ on a blind MDP induces a Puiseux Markov chain, which defines communicating classes.
The graph of communicating classes is a directed graph with one vertex per communicating class and an edge between two classes if the support of the exit distribution of one class contains a state in the other.
Formally, consider $G = (\V, \E)$ where $\V = \{ \B \subseteq \S : \B \text{ is a communicating class }\}$ and $(\B, \another{\B}) \in \E$ if and only if $\supp(\delta(\B, \cdot)) \cap \another{\B} \not = \emptyset$. The graph of communicating classes for \Cref{Example: Almost-Sure is not Limit-Sure}, induced by the Puiseux policy $\policy$, where $\policy_\eps(w) = 1 - \eps$ and $\policy_\eps(c) = \eps$, is illustrated in \Cref{Figure: Graph of communicating classes}.

\begin{figure}[ht]
\centering
\begin{tikzpicture}[
    node distance = 1cm and 1.5cm,
    every initial by arrow/.style = {thick},
    thick,
    ->,
    state/.style={circle, draw, minimum size=1cm}
  ]

  \node (B1) [state] {$B_1 = \{s_0\}$};
  \node (B2) [state, right = of B1] {$B_2 = \{s_1\}$};
  \node (B4) [state, right = of B2] {$B_4 = \{\top\}$};
  \node (B3) [state, left = of B1] {$B_3 = \{\bot\}$};

  \path
    (B1) edge[above] node[above] {} (B2)
    (B2) edge node[above] {} (B4);

\end{tikzpicture}
\caption{
	Graph of communicating classes induced by the Puiseux strategy $\policy$, where $\policy_\eps(w) = 1 - \eps$ and $\policy_\eps(c) = \eps$, for \Cref{Example: Almost-Sure is not Limit-Sure}. 
	Each node represents a communicating class in the induced Puiseux Markov chain. 
	Directed edges denote non-zero exit probabilities (in the limit $\varepsilon \to 0$) between classes. 
}
	\label{Figure: Graph of communicating classes}
\end{figure}
	
The following result shows that communicating classes have a laminar structure.

\begin{lemma}
    \label{Result: Laminar communicating classes}
    Consider a Puiseux Markov chain with a set of states $\S$ and disjoint communicating classes $\B_1, \B_2, \ldots, \B_k \subseteq \S$, with $k \ge 2$.
    We have that $\B \defas \sqcup_{i \in [k]} \B_i$ is a communicating class if and only if  $\{ \B_1, \B_2, \ldots, \B_k \}$ is a bottom strongly connected component in the graph with vertices $ \{ \B_1, \B_2, \ldots, \B_k \} \sqcup \{ \bot \}$ and edges 
    \[
    \{ (\B_i, \B_j) : \exists s \in \B_j, \, s \in \supp(\delta(\B_i, \cdot)) \}
    \sqcup \{ (\B_i, \bot) : \exists s \in \S \setminus \B, \, s \in \supp(\delta(\B_i, \cdot)) \} \,.
    \]
\end{lemma}

\begin{proof}[Proof sketch]
    The graph of the statement considers edges based on the exit distribution of communicating classes.
    On the one hand, if $\B$ is a communicating class, then starting in $\B$ the dynamic reaches every other state in $\B$ before exiting it. 
    In particular, the exit distribution connects communicating classes between each other without leading to states outside.
    On the other hand, if $\{ \B_1, \B_2, \ldots, \B_k \}$ is a bottom strongly connected component, then exit distributions link the classes, ensuring mutual reachability. 
    Therefore, starting in a state in $\B$ the dynamic reaches every other state in $\B$ before exiting it, so $\B$ is a communicating class.
\end{proof}

\begin{proof}
    Assume that $\B$ is a communicating class. 
    We show that $\B$ is a bottom strongly connected component in the graph of the statement.
    Consider $i \in [k]$ arbitrary. 
    We show that all edges of $\B_i$ lead to communicating classes in $\B$, i.e., $\supp( \delta(\B_i, \cdot) ) \subseteq \B$.
    By contradiction, assume that $\another{s} \in \supp( \delta(\B_i, \cdot) ) \cap \S \setminus \B$, equivalently, the graph of the statement contains an edge $(\B_i, \bot)$.
    Because $\B$ is a communicating class and $k \ge 2$, there exists $j \not = i$ and $\B_j$ such that $(\B_i, \B_j)$ is an edge in the graph of the statement.
    Consider states $s \in \B_i$ and $\another{\another{s}} \in \B_j$ where $\another{\another{s}}$ is such that $\delta( \B_i, \another{\another{s}}) > 0$.
    On the one hand, because $\B$ is a communicating class, $\another{s}$ is reached starting from $s$ before exiting $\B$, i.e., 
    \[
    \lim_{\eps \to 0^+} \PP_{s}^\eps( \exit(B) < \reach(\another{\another{s}}) ) = 0 \,.
    \]	
    On the other hand, by definition of exit distribution, there is a positive limit probability to exit $\B_i$ through $\another{s}$ which is outside of $\B$, i.e.,   
    \[
    \lim_{\eps \to 0^+} \PP_{s}^\eps( \Exit(B, \another{s})) \ge \delta( \B_i, \another{s} ) > 0 \,.
    \]
    This is a contradiction.
    Therefore, $\supp( \delta(\B_i, \cdot) ) \subseteq \B$.
    We are left with showing that $\B$ is strongly connected in the graph of the statement.
    
    Consider $i \not = j \in [k]$ arbitrary. 
    We show that $\B_i$ is connected to $\B_j$ in the graph of the statement.
    Consider $s \in \B_i$ and $\another{s} \in \B_j$.
    On the one hand, because $\B$ is a communicating class,  $\another{s}$ is reached starting from $s$ before exiting $\B$.
    On the other hand, the set of all reachable states from $s$ before having left $\B$ is characterized by the following procedure.
    Start including $s$.
    First, closure by communicating class, if $\another{\another{s}}$ is reachable from $s$ and $\another{\another{s}} \in \B_\ell$ with $\ell \in [k]$, then all states in $\B_\ell$ are included.
    Second, closure by exit distribution, if a state is in the support of the exit distribution of the reachable communicating classes, then it also is included.
    Repeat these closures until no more states are included to obtain the set of all reachable states from $s$ before having left $\B$.
    In particular, $\another{s}$ must be included in one of these two closures, and $\B_i$ is connected to $\B_j$ through, for example, a minimal sequence of additions in this process to include $\another{s}$ as a reachable state from $s$ before having left $\B$.
    We conclude that $\B$ is a bottom strongly connected component in the graph of the statement.
    
    Assume that $\B$ is a bottom strongly connected component in the graph of the statement. 
    We show that $\B$ is a communicating class.
    Consider $s, \another{s} \in \B$.
    We show that the limit probability of starting at $s$ and reaching $\another{s}$ before leaving $\B$ is $1$. 
    Consider $i, j \in [k]$ such that $s \in \B_i$ and $\another{s} \in \B_j$.
    By definition of communicating classes, transitions within a communicating class occur before exiting it, so it is sufficient to consider only the transitions exiting communicating classes.
    Because $\B$ is a bottom strongly connected component in the graph of the statement, the exit distribution of all communicating classes leads to states in $\B$. 
    In particular, the transitions between the communicating classes are taken infinitely more often than those exiting $\B$.
    Therefore, it is enough to show that starting from $\B_i$ the probability of having visited $\B_j$ after $k$ transitions between communicating classes is strictly positive.
    Denote the smallest positive exit probability $\nu \defas \min \{ \delta(\B_\ell, s) : \ell \in [k], s \in \S, \delta(\B_\ell, s) > 0 \} > 0$.
    Because $\B$ is strongly connected in the graph of the statement, there is a directed path between $\B_i$ and $\B_j$.
    Then, starting from $\B_i$ the probability of having visited $\B_j$ after $k$ transitions between communicating classes is at least $\nu^k > 0$. 
    We conclude that $\B$ is a communicating class.
\end{proof}

The laminar structure implies the following bound on the number of communicating classes.

\begin{corollary}
    \label{Result: Number of communicating classes}
    Consider a Puiseux Markov chain with a set of states $\S$.
    There are at most $2|\S| - 1$ communicating classes.
\end{corollary}

The following result characterizes reachable states in a Puiseux Markov chain.

\begin{lemma}
    \label{Result: Characterization of reachable states}
    Consider a Puiseux Markov chain with a set of states $\S$.
    The limit reachability is given by connectivity in the graph of communicating classes as follows.
    For all states $s, \another{s} \in \S$, we have that 
    \[
    \lim_{\eps \to 0^+}  \PP_{s}^\eps( \Reach(\another{s}) ) > 0 
    \]
    if and only if $\{ s \}$ is connected to a communicating class $\B \ni \another{s}$ in the graph of communicating classes.
\end{lemma}

\begin{proof}
    Note that the set of all reachable states in the limit from $s$, i.e., 
    \[
    \left\{ \another{s} \in \S : \lim_{\eps \to 0^+}  \PP_{s}^\eps( \Reach(\another{s}) ) > 0 \right \} \,,
    \] 
    is characterized as the outcome of the following procedure, which is similar to the one in the proof of \Cref{Result: Laminar communicating classes}.
    Start including $s$.		
    Start from $\{s\}$.
    First, closure by communicating class, if a state is included and this state is in some communicating class, then all states in the communicating class must be included.
    Second, closure by exit distribution, if a state is in the support of the exit distribution of a reachable communicating class, then it also must be included.
    Repeat the first and second closures until no more states are included to obtain the set of all reachable states in the limit from $s$.
    In particular, $\another{s}$ is reachable in the limit from $s$ if and only if the communicating class $\{s\}$ is connected to $\B \ni \another{s}$ through, for example, a minimal sequence of additions in this process to include $\another{s}$ as a reachable state from $s$.
\end{proof}

The following result characterizes limit-sure reachability in Puiseux Markov chains.

\begin{lemma}
    \label{Result: Characterization of limit-sure reachability}
    A Puiseux Markov chain, with a set of states $\S$ and a reachability objective, is limit-sure winning starting from $s \in \S$ if and only if, for all communicating classes $\B \subseteq \S$, if $\{s\}$ is connected to $\B$ in the graph of communicating classes, then $\B = \{ \top \}$ or the support of its exit distribution is not empty, i.e., $\supp( \delta (\B, \cdot) ) \not = \emptyset$.
\end{lemma}

\begin{proof}[Proof sketch]
    On the one hand, if $P$ is limit-sure winning, then, by \Cref{Result: Characterization of reachable states}, $\{s\}$ is connected to $\{ \top \}$ in the graph of communicating classes.
    By contradiction, if $\{s\}$ is connected to $\B$ and its exit distribution has empty support, then starting from $s$ the dynamic has positive probability of reaching and staying forever in $\B$, which contradicts the limit-sure winning property.
    On the other hand, if $s$ satisfies the assumptions, then we show that the dynamic eventually exits every subset of states containing $s$ and not $\top$. 
    Therefore, the dynamic reaches $\top$ with probability one in the limit, which proves that $P$ is limit-sure winning.
\end{proof}

\begin{proof}
    Consider a Puiseux Markov chain $P$ and a state $s \in \S$.
    Assume that $P$ is limit-sure winning starting from $s$.
    On the one hand,
    \[
    \lim_{\eps \to 0^+} \PP^{\sigma(\eps)}_{s}( \Reach(\top) ) = 1 > 0 \,.
    \]
    In particular, by \Cref{Result: Characterization of reachable states}, $\{ s \}$ is connected to a communicating class $\B \ni \top$.
    Note that, by definition, $\B = \{ \top \}$ is an absorbing communicating class.
    By \Cref{Result: Laminar communicating classes}, the only communicating class containing $\top$ is $\{ \top \}$.
    We conclude that $\{s\}$ is connected to $\{ \top \}$ in the graph of communicating classes.
    On the other hand, consider a communicating class $\B$ such that $\{s\}$ is connected to $\B$ in the graph of communicating classes.
    Consider $\another{s} \in \B$. 
    By \Cref{Result: Characterization of reachable states}, 
    \[
    \lim_{\eps \to 0^+}  \PP_{s}^\eps( \Reach(\another{s}) ) > 0 \,.
    \]
    By contradiction, if $\B$ is absorbing, then 
    \[
    \lim_{\eps \to 0^+} \PP^{\sigma(\eps)}_{s}( \Reach(\top) ) \le 1 - \lim_{\eps \to 0^+}  \PP_{s}^\eps( \Reach(\another{s}) ) < 1 \,,
    \]
    which is a contradiction.
    We conclude that the support of the exit distribution of $\B$ is not empty.
    
    Assume that, for all communicating classes $\B \subseteq \S$, if $\{s\}$ is connected to $\B$ in the graph of communicating classes, then $\B = \{ \top \}$ or the support of its exit distribution is not empty, i.e., $\supp( \delta (\B, \cdot) ) \not = \emptyset$.
    We show that $P$ is limit-sure winning starting from $s$.
    Note that, starting from $s$, the dynamic either reaches $\top$ or stays in a subset of $\S \setminus \{ \top \}$ forever.
    For a subset $\mathcal{C} \subseteq \S$, denote the time from which the dynamic never leaves $\mathcal{C}$ again by
    \[
    \stay(\mathcal{C}) \defas \min \{ n \ge 0 : \forall \another{n} \ge n \quad s_{\another{n}} \in \mathcal{C} \} \,.
    \] 
    We show that, for all $\mathcal{C} \subseteq \S$ such that $s \in \mathcal{C}$ and $\top \not \in \mathcal{C}$, the probability of staying in $\mathcal{C}$ forever is zero, i.e., 
    \[
    \lim_{\eps \to 0^+} \PP^{\sigma(\eps)}_{s}( \stay(\mathcal{C}) < \infty ) = 0 \,. 
    \]
    Fix an arbitrary $\mathcal{C} \subseteq \S$ such that $s \in \mathcal{C}$, $\top \not \in \mathcal{C}$, and 
    \[
    \lim_{\eps \to 0^+} \PP^{\sigma(\eps)}_{s}( \Reach(\mathcal{C}) ) > 0 \,.
    \]
    Take $\another{s} \in \mathcal{C}$ such that $\lim_{\eps \to 0^+} \PP^{\sigma(\eps)}_{s}( \Reach(\another{s}) ) > 0$.
    By \Cref{Result: Characterization of reachable states}, there are communicating classes $\B, \another{\B}$ such that $s \in \B$, $\another{s} \in \another{\B}$, and $\B$ is connected to $\another{\B}$ in the graph of communicating classes.
    If $\another{\B} \not \subseteq \mathcal{C}$, then, by definition of communicating classes, the dynamic exits $\mathcal{C}$ in finite time and derive the result. 
    By contradiction, assume that all communicating classes reachable from $\B$ are contained in $\mathcal{C}$.
    Consider the graph of communicating classes restricted to the classes included in $\mathcal{C}$.
    Because this is a directed subgraph, it has a bottom strongly connected component reachable from $\B$.
    Consider $\another{\B}$ the union of all states in this bottom strongly connected component.
    By \Cref{Result: Laminar communicating classes}, $\another{\B}$ is a communicating class.
    By construction, $\{ s_0 \}$ is connected to $\another{\B}$.
    Because $\another{\B} \subseteq \mathcal{C}$, we have that $\top \not \in \another{\B}$.
    Therefore, by assumption, the support of its exit distribution is not empty, i.e., $\supp( \delta (\B, \cdot) ) \not = \emptyset$.
    But this is a contradiction with being a bottom strongly connected component.
    We conclude that, for all $\mathcal{C} \subseteq \S \setminus \{ \top \}$, if $\lim_{\eps \to 0^+} \PP^{\sigma(\eps)}_{s}( \Reach(\mathcal{C}) ) > 0$, then $\lim_{\eps \to 0^+} \PP^{\sigma(\eps)}_{s}( \exit(\mathcal{C}) < \infty ) = 0$. 
    Therefore, $\lim_{\eps \to 0^+} \PP^{\sigma(\eps)}_{s}( \Reach(\top) ) = 1$ and the Puiseux Markov chain is limit-sure winning.
\end{proof}

\paragraph{Rank policy witness.}
Given a blind MDP $P$, we say that a Puiseux policy $\policy$ is a \emph{witness} for limit-sure winning if 
\[
\lim_{\eps \to 0^+} \PP^{\policy(\eps)}_{s_0} ( \Reach(\top) ) = 1\,.
\]
A (memoryless) Puiseux policy $\policy \colon [0, \eps_0) \to \Delta(\A)$ is a \emph{rank} policy if, for all $a \in \A$, the function $\eps \mapsto \policy(\eps)(a)$ is either constant to zero or an integer power of the identity up to normalization, i.e., there exists a function $f \colon \A \times [0, \eps_0) \to [0, 1]$ such that, for all $a \in \A$, there exists $i \ge 0$ such that $f(a, \eps) = \eps^i$, and $\policy(\eps)(a) = f(a, \eps) / \sum_{a \in \A} f(a, \eps)$.
In particular, for rank policies, we have that $\eps_0 = \infty$. 

The following result shows the existence of rank policy witnesses.

\begin{lemma}
    \label{Result: rank witness existence}
    Consider a blind MDP $P = (\S, \A, \delta, s_0)$ and a target state $\top$.
    Then, $P$ is limit-sure winning under memoryless policies if and only if there is a rank policy witness.
    Moreover, the description of the rank policy is of polynomial size.
\end{lemma}

\begin{proof}[Proof sketch]
    By \Cref{Result: function witness}, $P$ is limit-sure winning under memoryless policies if and only if there is a Puiseux policy witness.
    By \Cref{Result: Characterization of limit-sure reachability}, a Puiseux policy is a witness if and only if the respective graph of communicating classes has some properties.
    Note that the graph of communicating classes is defined only through the asymptotic behavior of the corresponding Puiseux Markov chain.
    Taking the communicating classes and the edges between them as a system of linear inequalities, we show the existence of a rank policy that induces the same graph of communicating classes and therefore is also a witness.
\end{proof}

\begin{proof}
    Consider a blind MDP $P = (\S, \A, \delta, s_0)$ and a target state $\top$.
    By \Cref{Result: function witness}, $P$ is limit-sure winning under memoryless policies if and only if there is a memoryless Puiseux function policy witness.
    In turn, by \Cref{Result: Characterization of limit-sure reachability}, a policy is a witness if and only if its graph of communicating classes satisfies some properties.
    We show that, if there is a Puiseux function policy whose graph satisfies these properties, then there is another polynomial-size rank policy that induces the same graph.
    
    Fix a Puiseux policy $\policy \colon [0, \eps_0) \to \Delta(\A)$ and consider the corresponding Puiseux Markov chain. 
    We claim that the graph of communicating classes is fully determined by the support of the exit distribution of communicating classes.
    Indeed, this graph can be constructed as follows.
    \begin{itemize}
        \item 
        Initialization. 
        Start by considering a communicating class for each singleton state.
        \item 
        Adding edges.
        Given the currently considered communicating classes, add all edges given by the support of their exit distribution.
        \item 
        Adding new communicating classes.
        Given the currently considered communicating classes, consider those that are not contained in another communicating class.
        With the support of their exit distribution, by \Cref{Result: Laminar communicating classes}, we find larger communicating classes if there are any.
    \end{itemize}
    By repeating the last two items until no other communicating class is found, we obtain the full graph of communicating classes.
    Therefore, this graph is fully determined by the support of the exit distribution of communicating classes.
    By \Cref{Result: characterization exit distribution}, the exit distribution of a communicating class is characterized in terms of exit graphs.
    Indeed, a state $\another{s}$ is in the support of the exit distribution of a communicating class $\B$ if and only if there exists a state $s \in \B$ and an exit graph $g$ in which $s$ can reach $\another{s}$, i.e., $g \in G_{\B}(s \to \another{s})$, such that the weight of $g$ is equal to the minimal weight of all exit graphs of $\B$, i.e., for all $\another{g} \in G_{\B}$ we have that $w(g) \le w(\another{g})$.
    We use this characterization to deduce the existence of a rank policy $\another{\policy}$ that induces the same graph as the policy $\policy$.
    
    Consider a parameterized function $f \colon \A \times [0, \eps_0) \to [0, 1]$, of parameters $(i(a))_{a \in \A}$, given by $f(a, \eps) = \eps^{i(a)}$.
    This function induces a policy $\another{\policy}(\eps)(a) = f(a, \eps) / \sum_{a \in \A} f(a, \eps)$, which in turn defines a Puiseux Markov chain with transitions
    \[
    \transition^\eps (s, \another{s}) = \sum_{a \in \A} \another{\policy}(\eps)(a) \transition(s, a)(\another{s}) = \frac{1}{\sum_{a \in \A} f(a, \eps)} \sum_{a \in \A} \eps^{i(a)} \transition(s, a)(\another{s}) \,.
    \]
    We show that there exist parameters $(i(a))_{a \in \A}$ such that the corresponding graph of communicating classes of $\another{\policy}$ coincides with the one given by $\policy$.
    Indeed, because the definition of exit distribution depends only on the weight of exit graphs, which is an asymptotic notion, we have that the weights of the policy $\policy$ induce a strategy with the same graph of communicating classes, i.e., defining
    \[
    i(a) \defas \inf \left\{ r \ge 0 : \lim_{\eps \to 0^+} \frac{\sum_{a \in \A} \policy(\eps)(a) \transition(s, a)(\another{s})}{\eps^r} \right\}
    \]
    we have that $\another{\policy}$ induces the same graph of communicating classes as $\policy$.
    We show that the parameters $(i(a))_{a \in \A}$ can be chosen to be integers of polynomial size.
    
    By the previous arguments, for simplicity and without loss of generality consider that the Puiseux policy $\policy$ is of the form $\policy(\eps)(a) = f(a, \eps) / \sum_{a \in \A} f(a, \eps)$, where $f \colon \A \times [0, \eps_0) \to [0, 1]$ is such that $f(a, \eps) = \eps^{i(a)}$.
    We construct a system of linear equations that is solved by $(i(a))_{a \in \A}$ and characterizes the induced graph of communicating classes.
    First, ranking of actions.
    Consider (strict) inequalities that characterize the order of $(i(a))_{a \in \A}$, i.e., (strict) inequalities of the form
    \[
    i(a) < i(\another{a}) 
    \qquad \text{or} \qquad 
    i(a) = i(\another{a}) \,.
    \]
    Second, support of exit distributions.
    Consider states $s, \another{s} \in \S$ and define the set actions that lead to the transition from $s$ to $\another{s}$ and are minimal in the ranking, i.e., 
    \[
    \mathcal{I}(s \to \another{s}) \defas \{ a \in \A : \delta(s, a)(\another{s}) > 0, \, \forall \another{a} \in \A, \, \delta(s, a)(\another{s}) > 0 \Rightarrow i(a) \le i(\another{a})  \} \,.
    \]
    Also, consider some selection and define $i(s \to \another{s}) \defas i(a)$, for some $a \in \mathcal{I}(s \to \another{s})$.
    Recalling that a state $\another{s}$ is in the support of the exit distribution of a communicating class $\B$ if and only if there exists a state $s \in \B$ and an exit graph $g$ containing the edge $(s, \another{s})$, i.e., $g \in G_{\B}(s \to \another{s})$, such that the weight of $g$ is equal to the minimal weight of all exit graphs of $\B$, i.e., for all $\another{g} \in G_{\B}$ we have that $w(g) \le w(\another{g})$.
    We write these restrictions as (strict) inequalities over $(i(a))_{a \in \A}$ by noticing that
    \[
    w(g) = \sum_{(s, \another{s}) \in g} i(s \to \another{s}) \,.
    \]
    Because there are finitely many communicating classes, and each of them has finitely many exit graphs, the graph of communicating classes induced by the policy $\another{\policy}$ is fully determined by a finite system of, possibly strict, inequalities over the variables $(i(a))_{a \in \A}$.
    These two sets of (strict) inequalities, along with positivity constraints, fully characterize the induced graph of communicating classes by $(i(a))_{a \in \A}$ in the following sense.
    Every solution of these inequalities $(i^*(a))_{a \in \A}$ define a function $f^* \colon \A \times [0, \eps_0) \to [0, 1]$, which defines a policy $\policy^*$.
    By the iterative construction of the graph of communicating class, the policy $\policy^*$ induces the same graph as $\policy$.
    We show that these inequalities have an integer solution $(i^*(a))_{a \in \A}$ of polynomial size.
    
    For a fixed order over $(i(a))_{a \in \A}$ and a selection defining $(i(s \to \another{s}))_{s, \another{s} \in \S}$, the inequalities considered are of the form
    \[
    i(a) \,\circ\, i(\another{a}),
    \qquad 
    \sum_{(s, \another{s}) \in g} i(s \to \another{s}) \,\circ\, \sum_{(s, \another{s}) \in \another{g}} i(s \to \another{s}),
    \qquad \text{or} \qquad 
    i(a) \ge 0 \,,
    \]
    where $\circ \in \{ <, \le \}$.
    Because this is a homogeneous system of equations, i.e., if $(i^*(a))_{a \in \A}$ is a solution, then, for all $\lambda > 0$, we have that $(\lambda  \cdot i^*(a))_{a \in \A}$ is also a solution, we can replace strict inequalities by inequalities that are not strict by adding $+ 1$ to the corresponding side of the inequality.
    Then, we consider a system of inequalities where $\circ$ is replaced by $\le$ or $\le 1 + $.
    Finally, we arrived at a system of linear equations constructed from the Puiseux policy $\policy$.
    Because this system of linear equations has a solution, it has a solution in the rational. 
    Moreover, the numerators and denominators of a solution can be bounded by Cramer's rule so they use polynomial size.
    Multiplying the rational solution of this system to obtain an integer solution of the original set of (strict) inequalities we finally conclude the existence of rank policy witnesses.
\end{proof}

\Cref{Result: rank witness existence} establishes the existence of a polynomial-size witness for limit-sure reachability of a blind MDP.
To prove the problem is in NP, the following result shows the existence of a polynomial time verifier that decides whether a rank policy is a witness of limit-sure reachability for a blind MDP or not.

\begin{lemma}
    \label{Result: Verifier existence}
    There exists a polynomial-time algorithm that, given a blind MDP and a rank policy, decides whether the rank policy is a witness of limit-sure reachability or not.
\end{lemma}

\begin{proof}[Sketch proof]
    The algorithm has two main steps.
    First, it constructs the graph of communicating classes.
    Second, it checks whether the only absorbing communicating class reachable from the initial state $s_0$ is $\{ \top \}$ or not.
    The graph is constructed iteratively as in the proof of \Cref{Result: rank witness existence}.
    The main operations are adding edges and communicating classes.
    Each of them take polynomial time and, by \Cref{Result: Number of communicating classes}, there are at most $(2 |\S| - 1)$ communicating classes.
    For the second step, by the characterization in \Cref{Result: Characterization of limit-sure reachability}, we can decide limit-sure reachability running a depth-first search starting at $\{s_0\}$.
\end{proof}

\begin{proof}
    Consider a blind MDP and a rank policy.
    The algorithm constructs the graph of communicating classes iteratively and checks whether the only absorbing communicating class reachable from the initial state $s_0$ is $\{ \top \}$ or not.
    Concretely, the algorithm constructs this graph similar to the proof of \Cref{Result: rank witness existence} as follows.
    \begin{itemize}
        \item 
        Initialization. 
        Start by considering a communicating class for each singleton state.
        \item 
        Adding edges.
        Given the currently considered communicating classes, add all edges given by the support of their exit distribution.
        \item 
        Adding new communicating classes.
        Given the currently considered communicating classes, consider those that are not contained in another communicating class.
        With the support of their exit distribution, by \Cref{Result: Laminar communicating classes}, we find larger communicating classes if there are any.
    \end{itemize}
    By repeating the last two items until no other communicating class is found, we obtain the full graph of communicating classes.
    We show that this algorithm runs in polynomial time.
    
    By \Cref{Result: Number of communicating classes}, there are at most $(2 |\S| - 1)$ communicating classes.
    There are two relevant operations for the algorithm that computes the graph of communicating classes.
    First, computing the support of the exit distribution of a communicating class.
    Second, checking whether a new communicating class should be added.
    We show how to perform these operations in polynomial time.
    
    Fix a communicating class $\B \subsetneq \S$. 
    We compute the support of its exit distributions, i.e., $\supp( \transition(\B, \cdot) )$.
    Consider $\another{s} \in \S \setminus \B$.
    By \Cref{Result: exit distribution} and \Cref{Result: characterization exit distribution}, the state $\another{s}$ is in the support of the exit distribution of $\B$ if and only if there exists a state $s \in \S$ and an exit graph $g \in G_{\B}(s \to \another{s})$ whose weight coincides with the smallest weight among all exit graphs in $G_{\B}$, i.e., the weight of some exit graph is $G^{\min}_{\B}$.
    We determine this in two steps.
    First, we compute the weight of some exit graph in $G^{\min}_{\B}$.
    Second, we compute the minimum weight over all exit graphs in $\cup_{s \in \S} G^{\min}_{\B}(s \to \another{s})$.
    Comparing these quantities we determine whether $\another{s}$ is in the support of the exit distribution of $\B$.
    For the first step, collapse all states in $(\S \setminus \B)$ into a single state and compute a minimal directed spanning tree where the weight of an edge is given by the leading power in the Puiseux power expansion of the corresponding transition.
    A directed spanning tree in this collapsed graph corresponds to an exit graph in the Puiseux Markov chain, and their weights coincide.
    By~\cite{gabow_efficient_1986}, this computation takes polynomial time in $|\S|$.
    This determines the weight of some exit graph in $G^{\min}_{\B}$.
    For the second step, for all $s \in \S$, we proceed similarly, i.e., we collapse all states in $\{ s \} \cup (\S \setminus \B)$ into a single state and compute a minimal directed spanning tree. 
    Then, we add the smallest weight among the transitions from $s$ to $\another{s}$.
    The result corresponds to the weight of an exit graph in $G^{\min}_{\B}$ containing the edge $(s, \another{s})$.
    By~\cite{gabow_efficient_1986}, this computation takes at polynomial time in $|\S|$, and we repeat it at most $|\B| \le |\S|$ times.
    Taking the minimum over all computed weights while varying $s \in \B$, we deduce the minimum weight of all graphs in $\cup_{s \in \S} G^{\min}_{\B}(s \to \another{s})$. 
    Recall that this weight coincides with the one of an exit graph in $G^{\min}_{\B}$ if and only if $\another{s}$ is in the support of the exit distribution.
    Therefore, comparing the weights obtained in the first and second steps we decide whether $\another{s}$ is in the support or not.
    
    Given all communicating classes computed so far and the support of their exit distributions, we check whether we can add another communicating class. 
    By the characterization in \Cref{Result: Laminar communicating classes}, this corresponds to computing bottom strongly connected components in a directed graph of at most $(2 |\S| - 1)$ vertices, which takes linear time by~\cite{dijkstra_discipline_1976}, and is done at most $(|\S| - 1)$ times.
    
    Finally, given the full graph of communicating classes induced by the rank policy, we check whether the policy is a witness of limit-sure reachability.
    By the characterization in \Cref{Result: Characterization of limit-sure reachability}, we run a depth-first search starting at $\{s_0\}$ and check whether the only reachable absorbing communicating class from $s_0$ is $\{ \top \}$.
    We conclude that checking whether the rank policy is a witness of limit-sure reachability or not takes polynomial time.
\end{proof}

\subsection{Lower Bound}

An NP-hardness result was established for a similar problem in \cite[Lemma 1]{chatterjee2013AutomatedAnalysisRealtime}, namely, it was shown that the problem of determining whether a two-player game with partial-observation with reachability objective is limit-sure winning under memoryless policies is NP-hard. 
The reduction constructed a game that is a directed acyclic graph, and replacing the adversarial player with a uniform distribution over choices shows that the limit-sure winning problem under memoryless policies in POMDPs is also NP-hard.

\begin{proposition}
    \label{Result: NP-hardness}
    For all constants $m \ge 0$, the problem of determining whether a POMDP $P$ with reachability objective is limit-sure winning under memory $m$ policies is NP-hard.
\end{proposition}

In the rest of the section, for completeness, we give an explicit reduction from 3-SAT~\cite{karp1972ReducibilityCombinatorialProblems} that prove \Cref{Result: NP-hardness}.

\paragraph{3-SAT.}
Consider boolean variables $x_1, x_2, \ldots, x_n$ and clauses $C_1, C_2, \ldots, C_m$ where each clause is the disjunction of three literals from the set $\{x_1, x_2, \ldots, x_n\} \cup \{\neg x_1, \neg x_2, \ldots, \neg x_n\}$.
The 3-SAT problem is determining if there is an assignment of the boolean variables that satisfy all clauses.

\begin{proof}[Proof sketch]
    The reduction is from 3-SAT, where each literal in each clause has a representative state.
    The main idea is to give one observation for all literals corresponding to the same boolean variable, so the controller must perform the same action in the states representing $x_i$ and $\neg x_i$. 
    The transitions for literals representing $x_i$ and $\neg x_i$ are different and represent their truth value by going to a target by different actions.
    Lastly, the starting state moves to every clause uniformly, so the limit-sure winning problem represents satisfying all clauses.
\end{proof}

\begin{proof}
    Consider an instance of 3-SAT given by boolean variables $x_1, x_2, \ldots, x_n$ and clauses $C_1, C_2, \ldots, C_m$.
    For $j \in [m]$, denote clause $C_j = \ell(j_1) \lor \ell(j_2) \lor \ell(j_3)$, where each literal $\ell \in \{x_1, x_2, \ldots, x_n\} \cup \{\neg x_1, \neg x_2, \ldots, \neg x_n\}$. 
    We construct the POMDP given as follows.
    \begin{itemize}
        \item 
        $\S \defas \sqcup_{j \in [m]} \{ \ell(j, 1), \ell(j, 2), \ell(j, 3) \} \cup \{ s_0, \top, \bot \}$, where $\ell(j, k)$ is a different state for each $j \in [m]$ and $k \in [3]$ that represents the $k$-th boolean variable of the $j$-th clause;
        \item 
        $\A \defas \{ t, f \}$ is the action set which represents assigning a truth value to a variable;
        \item 
        $\delta \colon \S \times \A \to \Delta(\S)$ is the transition function given by
        \[
        \delta(s, a) = \begin{cases}
            \frac{1}{m} \sum\limits_{j \in [m]} \1[\ell(j, 1)]
            & s = s_0 \\
            \1[\top]
            & a = t,  \exists j \in [m], k \in [3] : s = \ell(j, k) = x_i \\
            \1[\top]
            & a = f,  \exists j \in [m], k \in [3] : s = \ell(j, k) = \neg x_i \\				
            \1[\ell(j, k + 1)]
            & a = f,  \exists j \in [m], k \in [2] : s = \ell(j, k) = x_i \\
            \1[\ell(j, k + 1)]
            & a = t,  \exists j \in [m], k \in [2] : s = \ell(j, k) = \neg x_i \\
            \1[\bot]
            & a = f,  \exists j \in [m] : s = \ell(j, 3) = x_i \\
            \1[\bot]
            & a = t,  \exists j \in [m] : s = \ell(j, 3) = \neg x_i \\
            \1[s]
            & s \in \{ \top, \bot \}
        \end{cases}		
        \]
        In other words: $s_0$ moves to the first literal of each clause uniformly independent of the action; each literal moves to either $\top$ or the next literal in the clause; the terminal states $\top$ and $\bot$ are absorbing, i.e., for all actions $a \in \A$, we have that $\transition((s, a)) = \1[s]$, for $s \in \{\top, \bot\}$.
        \item 
        $\Z \defas \{ x_i : i \in [n] \} \cup \{ s_0, \top, \bot \}$ is the set of observations, one per boolean variable;
        \item 
        $o \colon \S \to \Z$ is the observation function that forces the controller to assign a consistent truth value to the literals and is given by
        \[
        o(s) = \begin{cases}
            s
            & s \in \{ s_0, \top, \bot\} \\
            x_i
            & \exists i \in [n], j \in [m], k \in [3] : s = \ell(j, k) \in \{ x_i, \neg x_i \}
        \end{cases}
        \]
        \item 
        $s_0 \in \S$ is the initial state.
    \end{itemize}
    See \Cref{Figure: 3-SAT reduction} for an illustration of this reduction.
    We show that this POMDP has reachability value $1$ if and only if the 3-SAT instance is satisfiable.
    
    \begin{figure}[t]
        \centering
        \begin{tikzpicture}[]
            \node (s0) at (0, 0) {$s_0$};
            \node (l11) at (-4, -2)  {$\ell(1, 1) = x_1$};
            \node (l12) at (-4, -4)  {$\ell(1, 2) = \neg x_2$};
            \node (l13) at (-4, -6) {$\ell(1, 3) = x_4$};
            \node (l21) at (4, -2) {$\ell(2, 1) = x_2$};
            \node (l22) at (4, -4) {$\ell(2, 2) = \neg x_3$};
            \node (l23) at (4, -6) {$\ell(2, 3) = \neg x_4$};
            \node (top) at (0, -4) {$\top$};
            \node (bot) at (0, -8) {$\bot$};
            
            \draw [->] (s0.south) -- (l11.north) node [midway, above] {$1 / m$};
            \draw [->] (s0.south) -- (l21.north) node [midway, above] {$1 / m$};
            \draw [->] (l11) -- (top) node [pos=0.1, below] {t}; 
            \draw [->] (l11.south) -- (l12.north) node [midway, left] {f};
            \draw [->] (l12) -- (top) node [pos=0.1, below] {f};
            \draw [->] (l12.south) -- (l13.north) node [pos=0.2, left] {t};
            \draw [->] (l13) -- (top) node [pos=0.2, below] {t};
            \draw [->] (l13) -- (bot) node [pos=0.1, below] {f};
            \draw [->] (l21) -- (top) node [pos=0.1, below] {t};
            \draw [->] (l21.south) -- (l22.north) node [midway, right] {f};
            \draw [->] (l22) -- (top) node [pos=0.1, below] {f};
            \draw [->] (l22.south) -- (l23.north) node [pos=0.2, right] {t};
            \draw [->] (l23) -- (top) node [pos=0.2, below] {f};
            \draw [->] (l23.south) -- (bot) node [pos=0.1, below] {t};
            
            \draw[dashed] (l12.east) -- (l21.west) node[midway, above] {$x_2$};
            \draw[dashed] (l13.east) -- (l23.west) node[midway, above] {$x_4$};
        \end{tikzpicture}
        \caption{Example of the reduction from 3-SAT to the limit-sure winning reachability problem in POMDPs for the 3-SAT instance $(x_1 \lor \neg x_2 \lor x_4 ) \land (x_2 \lor \neg x_3 \lor \neg x_4 )$.}
        \label{Figure: 3-SAT reduction}
    \end{figure}
    
    Assume the 3-SAT instance is satisfiable with a valuation $v \colon \{ x_i : i \in [n] \} \to \{ t, f \}$.
    Consider the memoryless deterministic policy for the controller given by any extension of the valuation, for example assigning action $t$ to states that do not represent a boolean variable. 
    In other words, consider $\sigma \colon \Z \to \A$ given by
    \[
    \sigma(s) = \begin{cases}
        v(x_i)
        & \exists i \in [n] : s = x_i \\
        t
        & s \in \{s_0, \top, \bot\}
    \end{cases}
    \] 
    We show that this policy guarantees a reachability probability of one, and therefore the POMDP has a reachability value of one.
    
    From the initial state $s_0$, any action leads uniformely to the set $\{ \ell(j, 1) : j \in [m] \}$.
    Therefore, it is enough to show that starting from $\ell(j, 1)$ we reach $\top$, for an arbitrary $j \in [m]$.
    Fix $j \in [m]$. 
    Since the clause $C_j = \ell(j_1) \lor \ell(j_2) \lor \ell(j_3)$ is satisfied by the valuation $v$, we consider three cases depending on the first literal that evaluates to true.
    \begin{enumerate}
        \item 
        Assume $v(\ell(j, 1)) = t$ and fix $i \in [n]$ such that $\ell(j, 1) \in \{ x_i, \neg x_i \}$. 
        Then, by the definition of $\sigma$ and $\delta$, we have that $\delta(\ell(j, 1), \sigma(x_i)) = \1[\top]$. 
        In other words, $\ell(j, 1)$ reaches $\top$ in a single transition.
        \item
        Assume $v(\ell(j, 1)) = f$ and $v(\ell(j, 2)) = t$ and fix $i_1, i_2 \in [n]$ such that $\ell(j, k) \in \{ x_{i_k}, \neg x_{i_k} \}$ for $k \in [2]$.
        Then, by the definition of $\sigma$ and $\delta$, we have that $\delta(\ell(j, 1), \sigma(x_{i_1})) = \1[\ell(j, 2)]$ and $\delta(\ell(j, 2), \sigma(x_{i_2})) = \1[\top]$. 
        In other words, $\ell(j, 1)$ reaches $\top$ after two transitions.
        \item
        Assume $v(\ell(j, 1)) = f$, $v(\ell(j, 2)) = f$, and $v(\ell(j, 3)) = t$ and fix $i_1, i_2, i_3 \in [n]$ such that $\ell(j, k) \in \{ x_{i_k}, \neg x_{i_k} \}$ for $k \in [3]$.
        Then, by the definition of $\sigma$ and $\delta$, we have that $\delta(\ell(j, 1), \sigma(x_{i_1})) = \1[\ell(j, 2)]$, $\delta(\ell(j, 2), \sigma(x_{i_2})) = \1[\ell(j, 3)]$, and $\delta(\ell(j, 3), \sigma(x_{i_3})) = \1[\top]$. 
        In other words, $\ell(j, 1)$ reaches $\top$ after three transitions.
    \end{enumerate}
    Since in all cases, starting from $\ell(j, 1)$ we reach $\top$, we have proven that $\sigma$ guarantees a reachability value of one.
    
    Assume the 3-SAT instance is not satisfiable, i.e., for all valuations $v \colon \{ x_i : i \in [n] \} \to \{ t, f \}$ there exists at least one clause that evaluates to false.
    We show that every deterministic memoryless policy leads to a reachability value strictly less than one and therefore the POMDP does not have a reachability value of one.
    
    Note that the reachability value of our POMDP may consider only deterministic policies.
    This observation holds POMDPs with general policies~\cite{feinberg1996MeasurabilityRepresentationStrategic, venel2016StrongUniformValue}, and we argue that this holds for our POMDP even when considering only memoryless policies because there are no loops in the dynamic.
    Indeed, our POMDP can be seen as an extended-form game with one player.
    Moreover, the controller may remember all of their previous actions of the game since, during the process, no observation is presented twice before reaching the states $\top$ and $\bot$, where the outcome is determined. 
    Therefore, Kuhn's theorem~\cite[Section 5]{aumann1964MixedBehaviorStrategies} applies and every memoryless policy induces the same distribution over outcomes that some distribution over deterministic policies.
    In other words, the reachability value of our POMDP may consider only deterministic policies.
    
    Consider a deterministic memoryless policy $\sigma \colon \Z \to \A$. 
    Define the valuation $v \colon \{ x_i : i \in [n] \} \to \{ t, f \}$ given by $v(x_i) = \sigma(x_i)$.
    Since the 3-SAT instance is not satisfiable, there exists a clause $C_j$ with $j \in [m]$ such that $v(C_j) = f$.
    Therefore, under $\sigma$, starting from $\ell(j, 1)$, the dynamic does not reach $\top$ by a similar argument as before.
    We conclude that the reachability probability starting from $s_0$ and following $\sigma$ is at most $1 - 1/m$, and therefore the POMDP does not have value one.
    This concludes the proof of NP-hardness.	
\end{proof}

We finish this section with the proof of \Cref{Result: NP-completeness}.

\begin{proof}[Proof of \Cref{Result: NP-completeness}]
    \Cref{Result: NP-hardness} establish the NP-hardness. 
    \Cref{Result: rank witness existence} and \Cref{Result: Verifier existence} imply the existence of a rank policy of polynomial size and a polynomial-time verifier, i.e., they imply the NP upper bound of the problem.
\end{proof}

\subsection{Extensions}

In this section, we discuss several extensions of \Cref{Result: NP-completeness}.
The following result shows that \Cref{Result: NP-completeness} extends to constant memory policies. 

\begin{corollary}
    \label{Result: constant memory NP-complete}
    The problem of determining whether a POMDP $P = (\S, \A, \delta, \Z, o, s_0)$ with reachability objective is limit-sure winning under constant memory policies is NP-complete.
\end{corollary}

\begin{proof}[Proof sketch]
    The NP-hardness follows from Proposition~1. 
    The NP upper bound is obtained as follows. 
    For an amount of memory $m \ge 1$, guessing the update function $\policy_u$, we can solve the memoryless problem on the product of the POMDP and the memory elements. 
    Hence, the NP upper bound follows.
\end{proof}

\begin{proof}
    The NP-hardness follows from \Cref{Result: NP-hardness}. 
    The NP upper bound is obtained as follows. 
    Consider an amount of memory $m \ge 1$.
    Note that an update function $\policy_u \colon [m] \times \Z \times \A \to [m]$ is a finite object of polynomial size.
    Moreover, fixing an update function $\policy_u$ and considering the product POMDP with states $\S \times [m]$, memoryless policies in the product correspond to policies with memory $m$ in the original POMDP.
    The formal definition of the product POMDP is $P_m \defas \left ( \S \times [m], \A, \another{\delta}, \Z \times [m], \another{o}, (s_0, 1) \right )$ where
    \begin{itemize}
        \item 
        the observation function $\another{o}$ is defined as, for all $s \in \S$ and $\mu \in [m]$, 
        \[
        \another{o} \left ( (s, \mu) \right ) \defas (o(s), \mu) \,,
        \]
        \item 
        the transition function $\another{\transition}$ is given by
        \[
        \another{\delta}((s, \mu), a)((\another{s}, \another{\mu})) 
        \defas \begin{cases}
            \delta(s, a)(\another{s}) 
            & \policy_u(\mu, o(s), a) = \another{\mu} \\
            0
            & \sim
        \end{cases}
        \]
    \end{itemize}
    Then, $P_m$ is limit-sure winning under memoryless policies if and only if $P$ is limit-sure winning under memory policies using memory amount $m$, which proves the NP upper bound. 
\end{proof}

\begin{remark}
    While \Cref{Result: constant memory NP-complete}  is stated for constant memory, the result holds for all memory bounds that are polynomial in the size of the POMDP, as this ensures that the witness is of polynomial size.
\end{remark}

While \Cref{Result: constant memory NP-complete} presents the extension to small memory policies, we further extend \Cref{Result: constant memory NP-complete} to other classic objectives, namely, parity or omega-regular objectives.
Parity objectives are canonical forms to express all $\omega$-regular properties~\cite{thomas1997LanguagesAutomataLogic}, e.g., all properties expressed in the linear-temporal logic (LTL) can be expressed as deterministic parity automata. 
In a parity condition, every state is labeled with a positive integer priority and the objective requires that the minimum priority visited infinitely often is even.
For any fixed memory policy, we obtain a Markov chain, and the recurrent classes are reached with probability~1. A recurrent class satisfies the parity objective with probability~$1$ if the minimum priority is even, which we refer to as a good recurrent class, otherwise satisfies the objective with probability~0. 
Hence, the limit-sure problem for parity under memoryless strategies reduces to limit-surely reaching the good recurrent classes.
Hence, the NP-completeness result of \Cref{Result: NP-completeness} and \Cref{Result: constant memory NP-complete} also extend to parity objectives. 
However, we focus on reachability objectives as all conceptual aspects are clarified in this basic and most fundamental objective.

\begin{corollary}
    \label{Result: constant memory parity NP-complete}
    The problem of determining whether a POMDP $P$ with parity objective is limit-sure winning under constant memory policies is NP-complete.
\end{corollary}

\begin{proof}[Proof sketch]
    The NP-hardness follows from Proposition~1. 
    The NP upper bound is obtained as follows. By guessing the support of rank policies, we can compute the recurrent classes, and the objective is to reach the good recurrent classes.  
    Hence, the NP upper bound follows.
\end{proof}

\begin{proof}
    The NP-hardness follows from \Cref{Result: NP-hardness}. 
    The NP upper bound is obtained as follows.
    First, consider memoryless policies because the proof for constant memory policies follows from the reduction described in the proof of \Cref{Result: constant memory NP-complete}. 
    
    Consider a POMDP that is limit-sure winning for the parity objective under memoryless policies.
    Then, by definition of the limit-sure winning property, for all $\eps > 0$, there exists a policy $\policy_\eps$ such that the probability of satisfying the parity condition is at least $1 - \eps$ under this policy. 
    Consider a sequence $(\eps_n \defas 1/n)_{n \ge 1}$ and a corresponding sequence of policies $(\policy_n)_{n \ge 1} \subseteq \Delta(\A)^\Z$. 
    Since the set of all possible supports per observation is finite, up to taking a subsequence, we assume that, for all $z \in \Z$, the support of $\policy_n(z)$ is invariant on $n$.
    For a Markov chain, the recurrent classes depend only on the support of the transition function.
    Therefore, the recurrent classes of the Markov chains induced by $\policy_n$ do not depend on $n$.
    
    Note that, in a Markov chain, the parity condition is satisfied if and only if a good recurrent class is reached.
    Therefore, the probability of satisfying the parity objective in the POMDP under $\sigma_n$ corresponds to the reachability probability to good recurrent classes (under $\sigma_n$).
    By guessing the support of a sequence of policies that are witness of the limit-sure property for the POMDP, we reduce the parity objective to the reachability objective of the corresponding good recurrent classes.
\end{proof}

\begin{remark}
    As mentioned before, while \Cref{Result: constant memory parity NP-complete}  is stated for constant memory, the result holds for all memory bounds that are polynomial in the size of the POMDP, as this ensures that the witness is of polynomial size.
\end{remark}

The following result shows that \Cref{Result: NP-completeness} extends to parametric Markov chains (pMCs) as presented by~\cite{junges2018FiniteStateControllersPOMDPs}. 

\begin{corollary}
	\label{Result: limit-sure pMCs NP-complete}
	The problem of determining whether a parametric Markov chain with reachability objective is limit-sure winning under constant memory policies is NP-complete.
\end{corollary}

\begin{proof}[Proof sketch]
	By~\cite[Corollary 1]{junges2018FiniteStateControllersPOMDPs}, the quantitative problem for POMDPs and pMCs for reach-avoid objectives are equivalent.
	In particular, the NP upper bound for limit-sure winning for reachability objectives translates immediately.
\end{proof}

\Cref{Result: limit-sure pMCs NP-complete} complements the qualitative objectives investigated by~\cite{junges2021ComplexityReachabilityParametric}, which include almost-sure but not limit-sure reachability.

\section{Conclusion and Future Work}

In this work, we presented the first solution for limit-sure winning with small memory policies for POMDPs. 
While the present work establishes the theoretical foundations, interesting directions of future work include the development of efficient encodings in NP-complete problems that have well developed solvers.
This includes SAT and Mixed Linear Program (MLP). 

Along evaluating combinations of reductions and solvers in standard benchmark instances, an important task is to identify classes of POMDPs on which a solver works particularly well, possibly including efficient heuristics for scalability and practical applications.
Besides classic reductions, incremental encodings should be investigated.
In other words, generating the clauses for SAT and the restrictions for MLP incrementally as opposed to generating all of them at the same time. 
Incremental encodings have been developed for almost-sure reachability, see for example~\cite{chatterjee2016SymbolicSATBasedAlgorithm}, and they take advantage of incremental solvers obtaining meaningful practical improvements.	

\paragraph{Acknowledgments.} 
This research was partially supported by Austrian Science Fund (FWF) 10.55776/COE12, the support of the French Agence Nationale de la Recherche (ANR) under reference ANR-21-CE40-0020 (CONVERGENCE project), and the ERC CoG 863818 (ForM-SMArt) grant.

\bibliographystyle{alpha}
\bibliography{refs}
	
\end{document}